\providecommand{\PaperFormat}{letterpaper}

\makeatletter
\@ifundefined{ReRun}{%

\@ifundefined{PaperFormat}{\newcommand{\PaperFormat}{letterpaper}}{}

\documentclass[11pt,\PaperFormat]{article}
\usepackage[margin=1in]{geometry}
\usepackage{times}
\usepackage{amssymb,amsfonts,amsmath,amsthm}
\usepackage{authblk}
\usepackage{color}
\usepackage{xspace}
\usepackage{url}
\newtheorem{theorem}{Theorem}[section]
\newtheorem{proposition}{Proposition}[section]
\newtheorem{lemma}{Lemma}[section]
\newtheorem{corollary}{Corollary}[section]
\newtheorem{claim}{Claim}[section]

\newcommand{\ch}{\mbox{\sf ch}}
\newcommand{\id}{\mbox{\sf id}}
\newcommand{\inp}{\mbox{\sf inp}}
\newcommand{\out}{\mbox{\sf out}}
\newcommand{\poly}{\mbox{poly}}
\newcommand{\caA}{{\cal A}}
\newcommand{\caB}{{\cal B}}
\newcommand{\caC}{{\cal C}}
\newcommand{\Otilde}{\widetilde O}
\newcommand{\oDelta}{\vec \Delta}
\newcommand{\eps}{\varepsilon}
\newcommand{\LOCAL}{{\sc local}\xspace}

\begin{document}
%%%%%%%%%%%%%%%%%%%%%%%%%%%%%%%%%%%%%%%%%%%%%%%%%%%%%

\title{\bf  Local Conflict Coloring}

\author[1,2]{Pierre Fraigniaud\protect\@ifundefined{ReRun}{\thanks{Additional support from ANR project DISPLEXITY.}}{\protect\footnotemark[1]}}
\author[1,2]{Marc Heinrich\protect\@ifundefined{ReRun}{\thanks{Supported by ENS Paris.}}{\protect\footnotemark[2]}}
\author[1,2]{Adrian Kosowski\protect\footnotemark[1]}

\affil[1]{\small CNRS and Paris Diderot University}
\affil[2]{\small GANG Project, Inria Paris}

\date{}

}{}

\providecommand{\InShort}{}
\providecommand{\InFull}{}

\renewcommand{\InShort}[1]{}
\renewcommand{\InFull}[1]{#1}

\@ifundefined{FullVersion}{%
	\renewcommand{\InShort}[1]{#1}
	\renewcommand{\InFull}[1]{}
}{}

\makeatother

\maketitle

\begin{abstract}
Locally finding a solution to \emph{symmetry-breaking} tasks such as vertex-coloring, edge-coloring, maximal matching, maximal independent set, etc., is a long-standing challenge in \emph{distributed network} computing. More recently, it has also become a challenge in the framework of \emph{centralized local} computation. We introduce \emph{conflict coloring} as a general symmetry-breaking task that includes all the aforementioned tasks as specific instantiations --- conflict coloring includes all \emph{locally checkable labeling}~tasks from [Naor\ \&\ Stockmeyer, STOC 1993]. Conflict coloring is characterized by two parameters $l$ and $d$, where the former measures the amount of freedom given to the nodes for selecting their colors, and the latter measures the number of constraints which colors of adjacent nodes are subject to.
 We show that, in the standard \LOCAL model for distributed network computing, if $l/d > \Delta$, then conflict coloring can be solved in $\Otilde(\sqrt{\Delta})+\log^*n$ rounds in $n$-node graphs with maximum degree~$\Delta$, where $\Otilde$ ignores the polylog factors in $\Delta$. The dependency in~$n$ is optimal, as a consequence of the $\Omega(\log^*n)$ lower bound by [Linial, SIAM J. Comp. 1992] for $(\Delta+1)$-coloring. An important special case of our  result is a significant improvement over the best known algorithm for  distributed $(\Delta+1)$-coloring due to [Barenboim, PODC 2015], which required $\Otilde(\Delta^{3/4})+\log^*n$ rounds. Improvements for other variants of coloring, including  $(\Delta+1)$-list-coloring, $(2\Delta-1)$-edge-coloring, coloring with forbidden color distances, etc., also follow from our general result on conflict coloring. Likewise, in the framework of centralized local computation algorithms (LCAs), our general result yields an LCA which requires a smaller number of probes than the previously best known algorithm for vertex-coloring, and works for a wide range of coloring problems.

\bigbreak

\end{abstract}

\vfill

\thispagestyle{empty}
\pagebreak
%\InShort{
\setcounter{page}{1}
%}

%%%%%%%%%%%%%%%%%%%%%%%%%%%%%%%%%%%%%%%%%%%%%%%%%%%%%
\section{Introduction}
%%%%%%%%%%%%%%%%%%%%%%%%%%%%%%%%%%%%%%%%%%%%%%%%%%%%%

\subsection{Context and Objective}

\emph{Distributed network computing} considers the computing model in which every node of a graph is an autonomous computing entity, and nodes exchange information by sending messages along the edges of the graph. In this context, \emph{symmetry breaking},  which is arguably the most important problem in distributed network computing has attracted a lot of attention, and several local forms of symmetry breaking tasks have been considered, including the construction of proper \emph{graph colorings}~\cite{BE09,BE13,Lin92,Nao91,PS92,SV93}, of \emph{maximal independent sets} (MIS)~\cite{ABI86,Lub86}, of \emph{maximal matchings}~\cite{Goo14,Hir12}, etc., to mention just a few. The main question in this framework is whether these tasks can be solved \emph{locally}, i.e., by exchanging data between nodes at short distance in the network. To tackle the locality issue, the complexity of a distributed algorithm is measured in term of number of \emph{rounds} in the  \LOCAL model~\cite{Pel00}, where a round consists in synchronously exchanging data along all the links of the network, and performing individual computations at each node. That is, a $t$-round algorithm is an algorithm in which every node exchanges data with nodes at distance at most $t$ (i.e., at most $t$ hops away) from it.

It is worth taking the example of coloring for understanding the computational challenges induced by the question of locality in distributed network computing. The main concern of distributed coloring is solving the \emph{$(\Delta+1)$-coloring task}, in which the nodes of a network $G$ are free to choose any color from the set $\{1,\dots,\Delta+1\}$, where  $\Delta$ is the maximum degree of $G$, as long as each node output a color that is different from all the colors output by its neighbors.\footnote{Solving $k$-coloring for $k<\Delta+1$ cannot be local, even if $G$ is $\Delta$-colorable, because the decision of a node can impact nodes far away from it, as witnessed by 2-coloring even cycles~\cite{Lin92}. } Several breakthroughs were almost simultaneously obtained towards the end of the 1980's. Awerbuch, Goldberg, Luby, and Plotkin~\cite{Awe89} devised a deterministic distributed $(\Delta+1)$-coloring algorithm running in a subpolynomial-in-$n$ number of rounds, which was subsequently improved by Panconesi and Srinivasan~\cite{PS92} to run in $2^{O(\sqrt{\log n})}$ rounds. Despite a quarter of a century of intensive research, this is still the best known distributed deterministic algorithm for $(\Delta+1)$-coloring in general graphs. Around the same time, Goldberg, Plotkin and Shannon~\cite{Gol88} and Linial~\cite{Lin92} designed distributed $(\Delta+1)$-coloring algorithms, performing in $O(\Delta^2+\log^*n)$ rounds, where $\log^*n$ denotes the least number of times the log-function should be applied on $n$ to get a value smaller than one\footnote{Formally, define $\log^{(0)}x= x$, and $\log^{(k+1)}x=\log\log^{(k)} x$ for $k\geq 0$; Then $\log^*x$ denotes the least integer $k$ such that $\log^{(k)}x<1$.}. These algorithms are significantly faster than the one in~\cite{PS92} for graphs with reasonably small maximum degree (e.g., $\Delta = O(\log^cn)$ for arbitrarily large constant $c>0$). Interestingly, the achieved dependence in $n$ is optimal for constant degree graphs, as \cite{Lin92} also proves that 3-coloring the $n$-node ring requires $\Omega(\log^*n)$ rounds, and this lower bound also holds for randomized algorithms~\cite{Nao91}. As a consequence, since Linial's contributions to $(\Delta+1)$-coloring, lots of effort has been devoted to decreasing the time dependence in $\Delta$ of coloring algorithms.%, and these efforts were indeed very fruitful.

Szegedy and Vishwanathan~\cite{SV93} show that a wide class of \emph{locally iterative} algorithms for $(\Delta+1)$-coloring must perform in $\Omega(\Delta\log\Delta)$ rounds, where an algorithm belongs to the locally iterative class if it has the property that, at each round, every node considers only its own current color together with the current colors of its neighbors, and updates its color value accordingly. This result was made more explicit by Kuhn and Wattenhofer~\cite{KW06}, who considered an almost identically defined model and proposed a locally iterative algorithm performing in $O(\Delta\log\Delta + \log^*n)$ rounds. Three years later, Barenboim and Elkin~\cite{BE09}, and Kuhn~\cite{Kuh09} independently proposed distributed $(\Delta+1)$-coloring  algorithms performing in  $O(\Delta+ \log^*n)$ rounds (see also~\cite{BEK14}). These latter algorithms are not iterative. Finally, Barenboim~\cite{Bar15} recently presented a distributed $(\Delta+1)$-coloring algorithm performing in $O(\Delta^{3/4} \log \Delta +\log^*n)$ rounds.

Other forms of coloring problems have also been tackled in the distributed network computing setting, including \emph{relaxations} of the classical vertex coloring problem, such as: edge-coloring, weak-coloring, defective coloring, vertex coloring with more than $(\Delta+1)$ colors, etc. (see, e.g., \cite{BE13} for a survey). In a number of practical scenarios, nodes aiming at breaking symmetry are also subject to more specific individual constraints. This is typically the case in frequency assignments in radio networks~\cite{GPT96,WL05}, in scheduling~\cite{Mar04}, and in digital signal processing~\cite{ZW03}, to mention just a few scenarios. In all these latter settings, each node $u$ is not initially free to choose any value from a color set $\caC$, but is a priori restricted to choose only from some subset $L(u)\subseteq \caC$ of colors. This framework is not captured by classical coloring, but rather by \emph{list-coloring}. As in the case of vertex coloring, distributed list-coloring can be approached from a locality perspective only if the lists satisfy $|L(u)|\geq \deg_G(u)+1$ for every node $u$ of a graph $G$ having degree~$\deg_G(u)$.

Vertex $(\Delta+1)$-coloring, as well as all of its previously mentioned relaxed variants, can be solved in $o(\Delta)+O(\log^*n)$ rounds~\cite{BE13}. However, the more complex task of $(\Delta+1)$-list-coloring was (prior to this work) only known to be  solvable in $\Otilde(|\caC|^{3/4})+O(\log^*n)$~\cite{Bar15} rounds, which is sublinear-in-$\Delta$ only for $|\caC|=o(\Delta^{4/3})$. Moreover, no sublinear (in $\Delta$) algorithms are known for MIS or maximal matching, for which the currently best algorithms run is $O(\Delta)+ \log^* n$ rounds~\cite{BE09,BE13,Kuh09}. (Again, the additional factor $\log^* n$  is unavoidable, and can be seen as an inherent cost of distributed symmetry breaking~\cite{Suo13}). In fact, there is evidence suggesting that no sublinear algorithms exist for these problems. For instance, for maximal matching, a time lower bound of $\Omega(\Delta + \log^* s)$ is known to hold for an anonymous variant of the \LOCAL model in which edges are equipped with locally unique identifiers from the range $\{1,\ldots,s\}$~\cite{Hir12}. In the standard \LOCAL model, a lower bound of $\Omega(\Delta)$ is known to hold for the \emph{fractional} variant of the maximal matching problem~\cite{Goo14}, while an $\Omega(\Delta/\log \Delta + \log^* n)$ lower bound holds for an extension of MIS called \emph{greedy coloring}~\cite{Gav09}.

In order to better understand which tasks can be solved in a number of rounds sublinear in~$\Delta$,  we focus on the general class of \emph{locally checkable labelings} (LCL) introduced by Naor and Stockmeyer~\cite{NS95}, which includes all tasks mentioned so far in this paper. Recall that a LCL is defined as a set of \emph{bad} labeled balls in graphs, where the ball of radius $r\geq 0$ centered at node $u$ in a graph $G$ is the subgraph of $G$ induced by all nodes at distance at most~$r$ from $u$ in $G$ (excluding edges between nodes at distance exactly $r$ from $u$), and where a label is assigned to each node. For instance, the bad balls for coloring are the balls of radius~1 in which the center node has the same label as one of its neighbors. Similarly, the bad balls for MIS are the balls of radius~1 for which either the center of the ball as well as one of its neighbors are both in the MIS, or none of the nodes in the ball are in the MIS. Every ball which is not bad is \emph{good}. To each LCL is associated a distributed task in which all nodes of an unlabeled graph $G$ must collectively compute a label at each node, such that all balls are good. Thus, our general objective is to tackle the following question:

\medskip
\centerline{\emph{What LCL tasks can be deterministically solved in $o(\Delta) + O(\log^* n)$ rounds?} }
\medskip

\noindent Given the state-of-the-art, we know since recently that answer to the above question is affirmative for $(\Delta+1)$-coloring~\cite{Bar15}, and there is also some very partial evidence hinting that this may not be true for MIS-type problems~\cite{Goo14,Gav09}. This also leads us to ask what makes $(\Delta+1)$-coloring and MIS so different? In the study of the randomized \LOCAL model, a separation in time complexity between $(\Delta+1)$-coloring and MIS has very recently been obtained by contrasting the randomized $(\Delta+1)$-coloring algorithms of Harris, Schneider, and Su~\cite{Har16} with lower bounds for MIS due to Kuhn, Moscibroda, and Wattenhoffer~\cite{KMW04}. However, this separation does not carry over directly to the deterministic setting. Here, in an attempt to advance understanding of the question for the deterministic scenario, we put forward the framework of~\emph{conflict coloring}, and show that efficient solutions to problems in the \LOCAL model can be obtained by taking advantage of their amenability to the conflict coloring framework.

\subsection{Our Results}

\paragraph{The setting.}

We define the general \emph{conflict coloring} task, which can be instantiated so as to correspond to any given LCL task. Roughly, conflict coloring is defined by a list of candidate colors given to each node (in the same spirit as list-coloring), and a list of conflicts between colors associated to each edge (following a convention used, e.g., when formulating unique games, CSP-s with binary conflict relations, etc.). For edge $\{u,v\}$, a conflict is a pair of the form $(c_u,c_v)$, indicating that a coloring where  $u$ has color $c_u$ and $v$ has color $c_v$ is illegal. Intuitively, given a LCL, the corresponding instance of conflict coloring is obtained by giving the list of all good balls centered at $u$ to every node $u$, and two balls given to adjacent nodes are in conflict whenever they are not consistent. Every LCL task is therefore a possible instantiation of conflict coloring (a given LCL task may have more than one conflict coloring representation). Note however that the power of conflict coloring extends beyond such a formulation of LCL tasks: depending on the instance, two colors in conflict along an edge $e$ do not, in general, need to be in conflict along another edge~$e'\neq e$.

We will speak of a conflict coloring with lists of length $l$ and conflict degree $d$, or more compactly of $(l,d)$-conflict-coloring, when all color lists given to the nodes are of length at least $l$, and for every edge $e$ and color $c$, the number of colors conflicting with color $c$ on edge $e$ does not exceed $d$. Intuitively, the larger the value of $l$, the easier the problem is, as every node has a choice among a large number of outputs. Conversely, the larger $d$ is, the harder the problem becomes as some nodes have to deal with many conflicts with at least one of their neighbors.

\paragraph{Distributed algorithm.}

Our main result is the design of a generic distributed algorithm which solves the conflict coloring task whenever $l/d>\Delta$ in graphs with maximum degree~$\Delta$. In the classical \LOCAL model for distributed network computing, our algorithm performs in $\Otilde(\sqrt{\Delta})+\log^*n$ rounds in $n$-node graphs, where the $\Otilde$ notation disregards polylogarithmic factors in $\Delta$.

The implications of our result are the following. There exists a trivial representation of $(\Delta+1)$-coloring as a conflict coloring task with $l/d\geq \Delta+1$. Therefore, our algorithm can be used to solve $(\Delta+1)$-coloring in $\Otilde(\sqrt{\Delta})+\log^*n$ rounds, which outperforms the currently fastest known $(\Delta+1)$-coloring algorithm by Barenboim~\cite{Bar15} performing in $\Otilde(\Delta^{3/4})+\log^*n$ rounds. In fact, for most classical variants of coloring, including $(2\Delta-1)$-edge-coloring, $(\Delta+1)$-list-coloring, coloring with forbidden color-distance sets~\cite{Rob91} given a sufficiently large palette, etc., our algorithm solves all these tasks in $\Otilde(\sqrt{\Delta})+\log^*n$ rounds, also improving the best  results known for each of them.  For small values of $\Delta$, our (deterministic)  algorithm for conflict coloring is even faster than the best known \emph{randomized} algorithms for $(\Delta+1)$-coloring~\cite{Har16}.

Interestingly, the bound $l/d>\Delta$  is essentially the best bound for which there exists a generic algorithm solving conflict coloring locally. Indeed, for every $l$ and $d$ such that   $l/d\leq \Delta$, there exists an instance of conflict coloring for which no solutions can be sequentially computed by a greedy algorithm selecting the nodes in arbitrary order. That is, the output of a node can impact the possible legal outputs of far away nodes in the network (like for $\Delta$-coloring~\cite{Lin92}). In particular, we are not aware of any instantiations of conflict coloring for MIS or maximal matching satisfying $l/d >\Delta$, which prevents us from solving these problems with a generic algorithm for conflict coloring. It might well be the case that there are no instantiation of conflict coloring for these problems  satisfying $l/d >\Delta$, which might be another hint that there are no algorithms running in $o(\Delta)+O(\log^*n)$ rounds for these tasks.

\paragraph{The techniques.}

From a technical point of view, the design of our algorithm required the development of a new technique, called a \emph{simplification mechanism}. This mechanism aims at iteratively reducing the difficulty of a given problem until it becomes simple enough to be trivially solved. More specifically, let $P_0$ be the problem we are aiming at solving. Our mechanism constructs a sequence $P_1,\dots,P_t$ of problems with the following three properties: (1) $P_{k+1}$ is ``easier'' to solve than $P_k$, and can be constructed from $P_k$ in $O(1)$ rounds, (2) $P_t$ is simple enough to be solved individually at each node, without any communication, and (3) given a solution to $P_{k+1}$, there is a $O(1)$-round algorithm computing a solution to $P_k$. Conflict coloring is perfectly suited to an application of the aforementioned simplification mechanism. Indeed, the set of colors in $P_0$ are those  in the lists given to the nodes in $G$. Constructing $P_{k+1}$ from $P_k$ increases the size of the lists (which is good), but the number of conflicts between colors also increases (which is bad). However, the increase rate of the number of conflicts will be shown to be lower than the increase rate of the size of the lists, which will eventually ensure that $P_t$ is easily solvable thanks to large lists, but a relatively small number of conflicts.

In conflict coloring, the main difficulty lies in obtaining a $O(\log^* n)$-round algorithm for solving an instance with ratio $l/d \ge 10\oDelta^2\ln \Delta$, given a graph with maximum degree $\Delta$ and edge orientation with maximum outdegree $\oDelta$. Subsequently, the conflict coloring problem then turns out to be directly amenable to an application of the \emph{arbdefective} coloring approach (cf.~\cite{BE10,BE13}), without having to resort to constructions of polynomials of the type used in~\cite{Bar15} during the recombination phase.
This is because the class of conflict coloring problems solved by our algorithm includes \emph{precoloring extension} (i.e., completing a partially given coloring of a graph), which can be handled directly through a modification of color lists available to vertices. By a careful (adaptive) choice of parameters of the arbdefective coloring, the complexity of our algorithm is reduced to $\Otilde(\sqrt \Delta)+ \log^* n$ rounds.

Disregarding polylogarithmic-in-$\Delta$ factors, the $\sqrt \Delta$-running time of our algorithm appears to be the limit of the precoloring extension technique, unless radically new algorithms are found to construct colorings in $O(\log^* n)$ rounds using significantly fewer colors than $\Otilde(\Delta^2)$. This latter problem has resisted all attempts for more than 20 years, since the publication of~\cite{Lin92}.

\paragraph{Additional results.}

Our result has also impact on \emph{centralized local} computation~\cite{ColoringCentralized,MMCentralized,MV13,NO08,MISCentralized}.  In this model, the local computation algorithm (LCA) is executed by a single computing unit which has access to the whole input graph, and needs to answer queries about a solution to the considered problem (e.g., ``is node $u$ in the MIS?''). For answering queries, the LCA probes the input graph, learning in each probe about some node $u$ and its neighborhood. The answers to the queries provided by the LCA must be consistent, that is, there must exist an implicit global solution that fits with the answers of the LCA. The complexity of such an algorithm is the number of probes that the LCA performs per query. Using our algorithm for conflict coloring, we show that there is a deterministic oblivious LCA for solving $(\Delta+1)$-list-coloring (and thus also $(\Delta + 1)$-coloring) using only $\Delta^{O(\sqrt{\Delta} \log^{5/2}\Delta)} \log^*n$ probes, improving the bound in~\cite{ColoringCentralized}.

\subsection{Other Related Work}

In addition to the aforementioned deterministic algorithms for $(\Delta+1)$-coloring, it is worth mentioning the \emph{randomized} algorithms for MIS in~\cite{ABI86,Lub86}, which both perform in $O(\log n)$ rounds, with high probability. Both algorithms can be transformed into randomized $(\Delta+1)$-coloring algorithms with the same round-complexity (e.g., using the reduction in~\cite{Lin92}). A ``direct'' randomized algorithm for $(\Delta+1)$-coloring with the same performances as these latter algorithms can be found in~\cite{BE13}. As a function of $\Delta$ and $n$, the best known randomized algorithms for $(\Delta+1)$-coloring, as well as for $(\Delta+1)$-list-coloring, perform in $O(\sqrt {\log \Delta}) + 2^{O(\sqrt{\log\log n})}$ rounds with high probability~\cite{Har16}. This result, combined with a previous lower bound on MIS of $\Omega(\log \Delta / \log \log \Delta)$ rounds~\cite{KMW04}, which also holds for randomized algorithms, implies a separation between the $(\Delta+1)$-coloring and MIS problems in the randomized case. On the positive side, MIS can be solved in $O(\log^2 \Delta) + 2^{O(\sqrt{\log\log n})}$ rounds with high probability~\cite{BEPS12}. We remark that the randomized and deterministic flavors of the \LOCAL model are significantly different, and in fact admit an exponential time separation, which has been recently shown for specific case of the problem of coloring a tree with~$\Delta$ colors~\cite{Cha16}. Whether a similar separation between randomized and deterministic complexity holds for MIS and the general $(\Delta+1)$-coloring problem is one of the main open questions of the field.

The list-coloring problem was introduced independently by Vizing~\cite{Viz76}, and Erd\"os, et al.~\cite{ERT79}. It is defined as follows. Let $G$ be a graph, let $\caC$ be a set of colors, and let $L:V\to 2^\caC$. If there exists a function $f:V\to \caC$ such that $f(v)\in L(v)$ for every  $v\in V(G)$, and $f(u)\neq f(v)$ for every $\{u,v\}\in E(G)$, then $G$ is said to be $L$-list-colorable. A graph is $k$-\emph{choosable}, or $k$-\emph{list-colorable}, if it has a list-coloring no matter how one assigns a list of $k$ colors to each node. The \emph{choosability number} $\ch(G)$ of a graph $G$ is the least number $k$ such that $G$ is $k$-choosable. Clearly, $\ch(G)\geq \chi(G)$, where $\chi(G)$ denotes the chromatic number of $G$. Computing the choosability number is actually believed to be harder than computing the chromatic number, because deciding the former is $\Pi_2^{\mbox{\sc p}}$-complete, while deciding the latter is NP-complete. In a distributed setting, $(\Delta+1)$-list-coloring is solvable in $O(\Delta+\log^*n)$ rounds by reduction to vertex-coloring. It was also recently proved to be solvable in time dependent on the size of the set of allowed colors, in $\Otilde(|\caC|^{3/4})+O(\log^*n)$ rounds~\cite{Bar15}.

It is also worth specifically mentioning the \emph{weak-coloring} problem, which asks for a coloring of  the nodes such that every non isolated node has at least one neighbor colored differently from it. It was proved in~\cite{NS95} that, in bounded-degree graphs with odd degrees, 2-weak-coloring can be solved in a constant number of rounds. This is one of the rare non-trivial distributed symmetry-breaking tasks that are known to be solvable in a constant number of rounds (in general, it is undecidable whether a solution to a locally checkable task can be constructed in constant time~\cite{NS95}). In graphs with constant maximum degree, for all locally checkable tasks, as well as their probabilistic extension~\cite{FKP13}, any randomized construction algorithm running in a constant number of rounds can be derandomized into a deterministic algorithm running in the same number of rounds~\cite{NS95,FF15}. However, this derandomization  result does not necessarily hold for randomized algorithms running in a \emph{non-constant} numbers of rounds. For example, it is not known whether there exists a deterministic $(\Delta+1)$-coloring algorithm running in a polylogarithmic number of rounds, or in other words, it is not known whether randomization helps for distributed $(\Delta+1)$-coloring.

Many other types of coloring have been considered in the literature, including using a larger number of colors, coloring the edges instead of the nodes, defective coloring, etc., and some of these tasks have been tackled in specific classes of graphs (planar, bounded arboricity, etc.). We refer to~\cite{BE13} for an excellent survey, also describing the various techniques of reductions between coloring, MIS, maximal matching, etc.

Regarding the centralized local model, essentially the same problems as for the distributed \LOCAL model have been studied, such as, e.g., maximal independent set \cite{MISCentralized}, and Maximum Matching \cite{MMCentralized}, for which algorithms were devised. A recent paper  \cite{ColoringCentralized} studies the relationship between the \LOCAL model and the centralized local computation model, including ways to adapt algorithms from the \LOCAL model to the centralized local setting. The resulting LCAs are deterministic and oblivious (they do not require to store information between queries), and, above all, they require a smaller number of probes than previously known algorithms. In particular, the method from  \cite{ColoringCentralized} yields a centralized $\Delta^2$-coloring LCA running in $O(\poly(\Delta) \cdot \log^* n)$ probes per query, and a centralized $(\Delta+1)$-coloring LCA running in $\Delta^{O(\Delta^2)} \cdot \log^* n$ probes per query.

%%%%%%%%%%%%%%%%%%%%%%%%%%%%%%%%%%%%%%%%%%%%%%%%%%%%%
\section{Model, Problem Setting, and Preliminaries}
%%%%%%%%%%%%%%%%%%%%%%%%%%%%%%%%%%%%%%%%%%%%%%%%%%%%%

\subsection{The \LOCAL Model}

We consider the usual framework for the analysis of locality in network computing, namely the \LOCAL model~\cite{Pel00}. In this model, a network is modeled as a \emph{connected} and \emph{simple} $n$-node graph (i.e., no loops, and no multiple edges). Each node $v$ of a network is given an \emph{identity}, denoted by $\id(v)$. This identity is a positive integer that is assumed to be encoded on $O(\log n)$ bits, and the identities of the nodes in the same network are pairwise distinct. In addition, every node $v$ may also be given an \emph{input} $\inp(v)\in\{0,1\}^*$.

For the sake of defining conflict coloring, we assume that the edges incident to a degree-$\delta$ node are identified by pairwise distinct labels in $\{1,\dots,\delta\}$, called \emph{port numbers}. No consistency between the port numbers at different nodes is assumed (in particular, an edge may have two different port numbers at its two extremities). Again, these port numbers are solely used for describing the input to every node in the context of conflict coloring, and provide no additional computing power to the  \LOCAL model (since nodes have identities).

In any execution of an algorithm $\caA$ in the \LOCAL model, all nodes  start at the same time. Initially, every node is only aware of its identity, and its input. As is usual in the framework of network computing, and for simplifying the description of the algorithm, we also assume that each node initially knows a polynomial upper bound on the total number $n$ of nodes. (See~\cite{KSV13} for techniques enabling to get rid of this assumption). Then all nodes perform a sequence of  \emph{synchronous rounds}. At each round, every node sends a message to its neighbors, receives the messages of its neighbors, and performs some individual computation. Which messages to send, and what computation to perform depend on the algorithm~$\caA$. The complexity of algorithm~$\caA$ in $n$-node graphs is the maximum, taken over all $n$-node graphs $G$, of the number of rounds performed by $\caA$ in $G$ until all nodes terminate.

Note that, whenever $t$ is known a priori, an algorithm $\caA$ performing in $t$ rounds can be simulated by an algorithm $\caB$ performing in two phases: First, in a network~$G$, every node $v$  collects all data from nodes at hop distance at most $t$ from $v$ (i.e., their identities, their inputs, as well as the structure of the connections between these nodes); Second, every node simulates the execution of $\caA$ in $B_G(v,t)$, where $B_G(v,t)$ is the \emph{ball} of radius $t$ around node~$v$ in graph $G$, that is, $B_G(v,t)$ is the subgraph of $G$ induced by all nodes at distance at most $t$ from $v$, excluding the edges between the nodes at distance exactly~$t$ from $v$.  Hence, the \LOCAL model enables to measure the \emph{locality} of a problem.

An algorithm satisfying the property that the output of every node is the same for all possible identity assignments to the nodes of the network is called \emph{identity-oblivious}, or \emph{ID-oblivious} for short.

\paragraph{Notation.}

We denote by $\deg_G(v)$ the \emph{degree} of a node $v$ in a graph $G$, that is the number of neighbors of $v$ in $G$, or, alternatively, the number of edges incident to $v$ in $G$ (recall that $G$ is a simple graph). We denote by $\Delta_G=\max_{v\in V(G)}\deg_G(v)$ the maximum degree of the nodes in $G$. The set of neighbors of node $v$ in graph $G$ is denoted by $N_G(v)$. Given an orientation of the edges of $G$, the set of out-neighbors of $v$ (nodes connected to $v$ by edges having their tail at $v$) is denoted by $\vec N_G(v)$, and the maximum node outdegree is denoted by $\oDelta_G$. When the graph $G$ is clear from the context, the index $G$ will be omitted from notation.

\subsection{Conflict Coloring}

Conflict coloring is defined as follows. Let $\caC$ be a finite set, whose elements are called colors. In graph $G$, each node $u\in V(G)$ is given as input
\begin{itemize}
\item a list $L(u)$ of colors in $\caC$, and
\item for every port number $i\in\{1,\dots,\deg_G(u)\}$, a list $C_i(u)=\big ((c_1,c'_1),\dots,(c_k,c'_k)\big)$ of conflicts, where $c_j\in L(u)$ and $c'_j\in\caC$ for every $j=1,\dots,k$.
\end{itemize}
To be well defined, the instance must satisfy the constraint that if $(c,c')\in C_i(u)$ and $u'$ is the neighbor of $u$ reachable from $u$ via port~$i$, then $(c',c)\in C_j(u')$, where $j$ is the port number of edge $\{u,u'\}$ at~$u'$. Each node $u$ in $G$ must output a color $\out(u)\in L(u)$ such that, for every edge $\{u,v\}$ with port number~$i$ at $u$, we have $(\out(u),\out(v))\notin C_i(u)$. That is, two adjacent nodes cannot be colored with a pair of colors that is indicated as a conflict for that edge. A given conflict coloring instance has conflict degree $d$ if, for all colors $c$, there are at most $d$ pairs of the form $(c, \cdot)$ in any of the lists $C_i(u)$. The conflict degree $d$ represents the maximum number of possible conflicts of one colors with other colors of one given neighbor.

For instance, $(\Delta+1)$-coloring is the instance of conflict coloring with $L(u)=\{1,\dots,\Delta+1\}$, and all conflict lists are of the form $(c,c)$ for all $c\in\{1,\dots,\Delta+1\}$. Expressing MIS as an instance of conflict coloring is not as straightforward. One way of doing this is the following. Assign lists of the form $L(u)=\{0,1\}\times\{1,\dots,\Delta\}$ to every node $u$. A color is thus a pair of integer values, where a color in the form of a pair $(1,i)$, for any $i \in \{1,\dots,\Delta\}$, is interpreted as ``$u \in \mbox{MIS}$'', and a color $(0,i)$ is likewise interpreted as ``$u\notin  \mbox{MIS}$, but the neighbor of $u$ reachable via port~$i$ belongs to the MIS''. We set a conflict along the edge from vertex $v$, following the $i$-th port to a neighboring vertex $u$, for all color pairs of the form $(1,j)$ at $v$ and $(1,k)$ at $u$, for all $j, k \in \{1,\dots,\Delta\}$, as well as for all color pairs of the form $(0, i)$ at $v$ and $(0,j)$ at $u$, for all $j \in \{1,\dots,\Delta\}$.

In fact, any LCL task can be expressed as an instance of conflict coloring. To see why, let us revisit MIS, and let us define MIS as an instance of conflict coloring in a brute force manner. One assigns $L(u)=\{S_1,\dots,S_{2^\delta}\}$ to every node $u$ of degree $\delta$, where $S_1$ is the $(\delta+1)$-node star with center labeled~1 and all leaves labeled~0, and, for $j>1$,  $S_j$ is a $(\delta+1)$-node star with center labeled~0, $\delta \choose x$ leaves labeled~1 for some $x\in\{1,\dots,\delta\}$, and all other leaves labeled~0. Conflicts in $C_i(u)$ are between incompatible stars $S_j$ at $u$ and $S'_k$ at $u'$ where the latter is the neighboring node of $u$ reachable from $u$ via port~$i$. More generally, any LCL task can be expressed as an instance of conflict coloring by assigning to every node $u$ a list of colors consisting of all good balls centered at $u$, and conflicts are between inconsistent balls between neighboring nodes.

For the sake of describing our algorithm, we define the \emph{conflict graph} $F$ associated to an instance of conflict coloring on $G=(V,E)$. The conflict graph $F$ is the simple undirected graph with vertex set
\[
V(F) = \{ (v, c) : v \in V(G), c \in L(v)\},
\vspace{-2mm}
\]
and edge set
\[
E(F) = \{\{(u,c), (u',c')\} : \mbox{$(c,c')\in C_i(u)$ where $i=$ port number of  $\{u,u'\} \in E(G)$ at node $u$} \}.
\]
In other words, to every edge $e=\{u,u'\} \in E(G)$  corresponds a bipartite graph with partitions $L(u)$ and $L(u')$, and there is an edge between a color $c\in L(u)$ and a color $c'\in L(u')$ if and only if these two colors are in conflict for edge $e$. For a conflict coloring in a graph of maximum degree $\Delta$, and conflict degree $d$, the conflict graph has degree at most $d \Delta$.

Let us note that, in conflict coloring, there is an interplay between the size, $l$, of the lists of available colors at each node (the larger the better as far as solving the task is concerned), and the conflict degree, $d$, of the colors along each edge.%, i.e., $d$ is the maximum degree of the conflict graph $F$ (the smaller the better).
We define \emph{$(l,d)$-conflict coloring} as conflict coloring with all lists of size $l$, and the degree of the conflict graph is at most~$d$. In the rest of the paper, we shall show that if the ratio between these two quantities is large enough, namely $l/d>\Delta$, then $(l,d)$-conflict coloring in solvable in a sublinear (in $\Delta$) number of rounds. For instance, $(\Delta+1)$-list-coloring corresponds to $l=\Delta+1$, and $d=1$, hence the ratio $l/d$ is sufficient to be covered by our approach. By contrast, for the previously described representation of MIS as conflict coloring, we have $l=2\Delta$ and $d=\Delta$, hence $l/d=2$.\footnote{A simple argument illustrating that $l/d=2$ is essentially the best ratio which can be achieved when using natural conflict-coloring-based representations of MIS is given in Appendix~A.}

%%%%%%%%%%%%%%%%%%%%%%%%%%%%%%%%%%%%%%%%%%%%%%%%%%%%%
% Technical sections
%%%%%%%%%%%%%%%%%%%%%%%%%%%%%%%%%%%%%%%%%%%%%%%%%%%%%

\subsection{Organization and Proof Outline}

In Sections~\ref{sec:deltasquare} and~\ref{sec:main} we provide the techniques and algorithms for solving $(l,d)$-conflict coloring for $l/d > \Delta$. Section~\ref{sec:deltasquare} lays out the main ingredient, namely, a routine for conflict coloring in $O(\log^* n)$ rounds when $l/d \ge 10 \Delta^2 \ln \Delta$ in a graph of maximum degree $\Delta$, or more generally when $l/d \ge 10 \oDelta^2 \ln \Delta$ and an orientation of the edges of the graph with outdegree $\oDelta$ is given. This is achieved by an application of our instance simplification technique, since the existence of color lists in the problem description precludes the application of simpler color reduction mechanisms (e.g., of the sort used by Linial~\cite{Lin92} for $\Delta^2$-vertex coloring). In Section~\ref{sec:main} we then solve any conflict coloring problem with $l/d > \Delta$ by applying the routines from Section~\ref{sec:deltasquare} on specific vertex-disjoint oriented subgraphs of $G$. These subgraphs are carefully constructed using the technique of arbdefective coloring~\cite{Bar15}, in such a way as to have sufficiently small outdegree $\beta$ for the condition $l/d \ge  10 \beta^2 \ln \Delta$ to hold within them. Finally, in Section~\ref{sec:centralized} we discuss implications of our conflict coloring routines for centralized LCAs, both in the case of $l/d \ge 10 \Delta^2 \ln \Delta$ and $l/d > \Delta$.

\section{Instance Simplification}\label{sec:deltasquare}

Our simplification mechanism, which allows us to generate progressively easier conflict coloring problems on a graph $G$, is now captured by the following key lemma. We will apply it to ``simplify'' a $(l,d)$-conflict-coloring problem $P = P_0$, such that $l / d \ge 10 \oDelta^2 \ln \Delta$, into one with a larger ratio $l/d$.

\begin{lemma}
\label{lem:simple}
Let $G$ be a graph with maximum degree $\Delta$ and a given edge orientation of outdegree at most $\oDelta$. Let $P_i$ be an instance of a $(l_i, d_i)$-conflict-coloring problem on graph $G$.
Then, for some integers $l_{i+1}, d_{i+1}$, there exists an instance $P_{i+1}$ of $(l_{i+1}, d_{i+1})$-conflict-coloring on graph $G$, such that:
\begin{enumerate}
\item There exists an ID-oblivious single-round local distributed algorithm which, given the \emph{input} of each node in $P_i$, outputs for each node its input in $P_{i+1}$.
\item There exists an ID-oblivious single-round local distributed algorithm which, given any valid \emph{output} of each node in $P_{i+1}$, outputs for each node a valid output for $P_{i}$.
\item The following condition is fulfilled for any $\eps>0$, when $\oDelta$ is larger than a sufficiently large constant:
$$\frac {l_{i+1}}{d_{i+1}} > \frac 1 {\Delta} \exp\left(\frac{1}{(e^2 + \eps)\oDelta^2}\frac{l_i}{d_i}\right).$$
\end{enumerate}
\end{lemma}

\smallskip
\noindent\emph{(For improved readability, the proof of the Lemma is postponed to Appendix~B at the end of the paper.)}

\smallskip
The construction used in the proof of Lemma~\ref{lem:simple} is the most technically involved part of our paper. Since the values of both $l$ and $d$ change in the steps of simplification mechanism, we inherently exploit the properties of conflict coloring, displaying that for our purposes, the class of conflict coloring problems needs to be addressed in its full generality of formulation. Indeed, even if the original problem $P_0$ is chosen as a relatively simple task, such as a list coloring problem (with $d=1$), all the subsequent problems $P_i$, $i\geq1$, which appear later on in the scheme, are of more general conflict form (with $d>1$).

The following lemma provides a criterion which allows us to determine the necessary number of iterations of the proposed simplification mechanism. It states that we can solve a $(l,d)$-conflict-coloring problem directly, without any further communication, given that the ratio $l/d$ is sufficiently large, subject to some additional assumptions constraining the structure of the input instance.
This is achieved through a greedy assignment of colors for the sufficiently simplified problem instance.

\begin{lemma}
\label{lem:greedy}
Consider an instance of the $(l,d)$-conflict coloring problem on a graph $G$ with maximum degree $\Delta$, such that the list of colors available to all nodes is $\{1,\ldots,l\}$. Suppose the following information available to all nodes:
\begin{itemize}
  \item Each node $v \in V$ receives its input $\inp(v)$ for the corresponding $(l,d)$-conflict coloring instance $P$ for $v$, accompanied by an integer label $\lambda(v) \in \{1,\ldots, s\}$, such that $\lambda(V)$ is a $s$-vertex-coloring  of the graph (i.e., $\lambda(u) \neq \lambda(v)$ for all $\{u,v\} \in E(G)$),
  \item A promise is given to all nodes $v\in V$ that $\inp(v) \in I$, where $I$ is a set known to all nodes.
\end{itemize}
If $\frac{l}{d} > \Delta s |I|$, then a solution to $P$ can be found in a local manner without communication (in $0$ rounds).
\end{lemma}

The proof of the lemma relies on the observation that all nodes can use their shared knowledge of set $I$ to determine an assignment of non-conflicting colors to each possible input from $I$, without communication.

\begin{proof}
Let $I' = I \times \{1,\ldots,s\}$. We construct the $(l,d)$-conflict-coloring as a simple function $c : I' \to \{1,\ldots,l\}$, where the color $\out(v)$ of a node $v$ with input $\inp(v)$ and label $\lambda(v)$ is given as $\out(v) = c (\inp'(v))$, where $\inp'(v) = (\inp(v),\lambda(v))$. The function $c$ is decided locally, by each node in an identical way, based only on knowledge of $I'$. To do this, we consider a fixed enumeration $I' = (\inp'_1, \inp'_2, \ldots, \inp'_{|I'|})$ of set $I'$. For $i \in \{1, \ldots, |I'|\}$, $\sigma \in \{1,\ldots, \Delta\}$, $\gamma \in \{1, \ldots, l\}$, let $S_{i,\sigma}(\gamma) \subseteq \{1,\ldots,l\}$ be the set of colors defined in $\inp'_i$ as forbidden  in a solution to $P$ for a vertex initialized with $\inp'_i$, given that the $\sigma$-th neighbor of this vertex obtains color $\gamma$. Notice that, by the conflict degree constraint for problem $P$, we have $|S_{i,\sigma}(\gamma)| \leq d$, for all $i, \sigma, \gamma$.
We now define function $c$ over input set $I'$ sequentially and greedily, fixing for successive $i = 1,\ldots, |I'|$ the value $c(\inp'_i)$ as the first (smallest) color value which can be assigned to a vertex having input $\inp'_i$ without causing a conflict with any potentially neighboring vertex which has already been colored, i.e., which has input $\inp'_j$, for some $j<i$:
\begin{equation}
\label{eq:iset}
c(\inp'_i) := \min \left(\{1, \ldots, l\} \setminus \bigcup_{j <i, 1\leq \sigma \leq \Delta} S_{i,\sigma}(c(\inp'_j))\right)
\end{equation}
Since each of the $|I'|$ possible input configurations conflicts with at most $d$ colors of its neighbor, for each of its $\Delta$ possible placements, and $l > d \Delta |I'| = d\Delta s |I|$ by assumption, it follows that using the rule~\eqref{eq:iset} we can assign a color for all feasible inputs and labels of nodes without running out of colors. We also remark that, for all $\{u,v\}\in E(G)$, we have $\lambda(u)\neq \lambda(v)$ by assumption, hence $\inp'(u) = (\inp(u),\lambda(u)) \neq (\inp(v), \lambda(v)) = \inp'(v)$. The correctness of the obtained conflict coloring $\out (v) = c(\inp'(v))$ now follows directly from the definition of function $c$ (cf.~Eq.~\eqref{eq:iset}).
\end{proof}

We can now combine the claims of Lemma~\ref{lem:simple} and Lemma~\ref{lem:greedy} to show that any conflict coloring problem $P_0$, given a sufficiently large initial ratio $l_0/d_0$, will after a small number of rounds be simplified by iterated application of Lemma~\ref{lem:simple} into a conflict coloring problem $P_t$, which is solvable without communication in view of Lemma~\ref{lem:greedy}. This leads us to the main technical lemma of this Section.

\begin{lemma}\label{lem:techlemma}
For a graph $G$ with maximum degree $\Delta$, a $s$-coloring of the vertex set, and a given orientation of edges with maximum outdegree $\oDelta$, where $\oDelta$ is at least a sufficiently large constant, any instance of the $(l,d)$-conflict coloring problem with $\frac{l}{d} \geq 10 \oDelta^2 \ln \Delta$ can be solved with a local distributed algorithm in at most $3 (\log^* \max\{s, l, \Delta\} - \log^* \frac{l}{d}) + 10$ rounds. In particular, the number of rounds of the algorithm can be written as $O(\log^* s + \log^* \Delta + \log^* d)$, where to obtain this bound we restrict excessively long color lists, so that $l_0 = d_0 \lceil 10 \Delta^2 \ln \Delta\rceil$.
\end{lemma}

\begin{proof}
To allow for a more compact write-up, we do not optimize the exact values of constants in the analysis. (In fact, the condition of the lemma can also be strengthened to $\frac{l}{d} \geq (e^2+2+ \eps') \oDelta^2 \ln \Delta$, for any $\eps'>0$, where $e^2 + 2 \approx 9.39$.)

Throughout the proof, we will assume that $\oDelta$ is sufficiently large that Clause~3 of Lemma~\ref{lem:simple} holds for the considered $(l,d)$-coloring problem with parameter $\eps = 0.1$. %Indeed, if $\oDelta = O(1)$, we can construct a solution to the considered conflict coloring problem directly in $O(\log^* s) + O(\oDelta^2) = O(\log^* s)$ rounds, by first obtaining an $O(\oDelta^2)$-coloring of the graph from its initial $s$-coloring using a variant of Linial's algorithm in $O(\log^* s)$ rounds (a procedure referred to as Arb-Linial, cf.~\cite{BE08}), and then greedily assigning in each round colors to all vertices of successive independent sets, corresponding to color classes of the given $O(\oDelta^2)$-coloring of $G$.

Now, let $P_0$ be the initially considered $(l,d)$-coloring problem ($l_0 = l$, $d_0 = d$). By iterating the simplification procedure from Lemma~\ref{lem:simple} in successive rounds, we obtain a sequence of problems $P_i$ with a rapidly increasing ratio $\frac{l_i}{d_i}$. Indeed, by applying Lemma~\ref{lem:simple} with $\eps = 0.1$, we have for sufficiently large $\oDelta$:
\begin{align*}
\frac {l_1}{d_1} &> \exp\left(\frac{1}{(e^2 + 0.1)\oDelta^2}10 \oDelta^2 \ln \Delta - \ln \Delta\right)\\  &> \exp\left(2.51\ln \Delta\right) > \Delta^{2.5} \ln \Delta \geq \oDelta^{2.5} \ln \Delta.
\end{align*}
Moreover, whenever $\frac {l_i}{d_i}\geq \oDelta^{2.5} \ln \Delta$, we have: $\oDelta \leq \left(\frac {l_i}{d_i \log \Delta}\right)^{2/5}$, hence $\frac{l_i}{d_i\oDelta^2} \geq \left(\frac{l_i}{d_i}\right)^{1/5} \log \Delta^{4/5}$. We obtain for sufficiently large $\oDelta$:
\begin{align*}
\frac {l_{i+1}}{d_{i+1}} &> \exp\left(\frac{1}{(e^2 + 0.1)\oDelta^2}\frac{l_i}{d_i} - \ln \Delta\right) \\ &> \exp\left(\frac{1}{(e^2 + 0.1)\oDelta^2}\frac{l_i}{d_i} - \frac 1 {\oDelta^{2.5}}\frac{l_i}{d_i} \right)\\ &> \exp\left(\left(\frac{l_i}{d_i}\right)^{0.2}\right).
\end{align*}
By an application of the above inequality over two successive steps, it follows that for all $i\geq 1$ the following condition:
$$
\frac {l_{i+2}}{d_{i+2}} > \exp\left(\frac {l_{i}}{d_{i}}\right),
$$
holds when $\oDelta$ is sufficiently large (we require $l_1 / d_1 > x$ to hold, where $x$ is the solution to the equality $\exp[x^{0.2}] = x^5$; we have $x \approx 2.45 \cdot 10^{10}$, and we recall that $l_1 / d_1 > \oDelta^{2.5}$). Thus, we have:
\begin{equation}
\label{eq:star}
\log^*\frac{l_t}{d_t} - \log^*\frac{l_0}{d_0} \geq \frac{t}{2}-1,\quad \text{for all $t\geq 0$}.
\end{equation}
We will now focus on finding a value of $t$ such that Lemma~\ref{lem:greedy} can be applied to problem $P_t$.

In order to bound the size of the set $I$ of feasible inputs for problems $P_t$ in our sequence, we will assume that the initial $(l_0,d_0)$-coloring problem $P_0$ is presented in \emph{standard interval form}, i.e., so that the color lists of each vertex $v\in V$ are identified with the set of consecutive integers, $L(v) = \{1,\ldots,l_0\}$. Should the initial color lists be of different form, a relabeling of colors by all nodes to obtain standard interval form can be performed in one computational round, preserving the conflict graph $F_0$ up to isomorphism. Then, for a node $v\in V$, its input in $P_0$ consists of a subset of forbidden color pairs from $\{1,\ldots,l\}^2$, assigned to each of the ports incident to $v$. By considering all possible input configurations, for given $l_0$ and $\Delta$ we define a set $I_0$ of feasible input configurations of problem $P_0$, obtaining:
$$
|I_0| \leq 2^{\Delta l_0^2}.
$$
Set $I_0$ can be computed locally (without communication) by all nodes of the graph.

By iteratively applying Clause~1 of Lemma~\ref{lem:simple}, we obtain that the input for a node $v$ in problem $P_{t}$ can be constructed by a $t$-round distributed ID-oblivious algorithm, using only the inputs of nodes in problem $P_0$ within a radius-$t$ ball around $v$ in graph $G$. Given $l_0$ and $\Delta$, by considering all possible topologies of a radius-$t$ ball of the graph and considering all possible inputs of $P_0$ for nodes within this ball, each node can compute without communication a set $I_t$ of feasible problem inputs for problem $P_t$. Since a radius-$t$ ball in $G$ contains fewer than $2\Delta^t$ nodes, we obtain a rough bound on the size of set $I_t$:
\begin{equation}
\label{eq:it}
|I_t| <  |I_0|^{2 \Delta^t} < 2^{2 l_0^2 \Delta^{t+1}}
\end{equation}
Now we find a value of $t$ for which the assumption $\frac{l_t}{d_t} > \Delta s |I_t|$ of Lemma~\ref{lem:greedy} is met. Taking into account Eq.~\eqref{eq:star}, it suffices to choose any value of $t$ which satisfies:
\begin{equation}
\label{eq:t4}
t \geq 2 \log^* (\Delta s |I_t|) - 2 \log^* \frac{l_0}{d_0} + 1.
\end{equation}
Moreover, in view of Eq.~\eqref{eq:it}, we have:
\begin{equation}
\label{eq:t5}
2\log^* (\Delta s |I_t|) \leq 2\log^*\left(\Delta s (2 l_0^2 \Delta^{t+1})\right) \leq 2\log^* \max\{s, l_0, \Delta\} + 2\log^* t + 8.
\end{equation}
Taking into account~\eqref{eq:t5}, by a very rough bound, condition~\eqref{eq:t4} is thus fulfilled for a suitably chosen value $t = 3 (\log^* \max\{s, l_0, \Delta\} - \log^* \frac{l_0}{d_0}) + 10$. In particular, we have $t = O(\log^* s + \log^* \Delta + \log^* d_0)$. For this value $t$, we can solve problem $P_t$ in zero rounds by Lemma~\ref{lem:greedy} as long as it is represented in standard interval form (with color lists $\{1,\ldots,l_t\}$ for each vertex); obtaining such a formulation requires one communication round.

Overall, we obtain an algorithm for solving the $(l_0,d_0)$-conflict-coloring instance $P_0$ in $O(t)$ rounds, by constructing an instance of $P_t$ from $P_0$ in $t$ rounds through $t$-fold application of Lemma~\ref{lem:simple}, solving problem $P_t$ in its standard interval form using Lemma~\ref{lem:greedy}, and eventually obtaining a solution to the original instance $P_0$ after a further $t$ rounds again in view of Lemma~\ref{lem:simple}. This completes the proof of the Lemma.
\end{proof}

Lemma~\ref{lem:techlemma} can be applied to solve $(l,d)$-conflict coloring on any graph $G$, using a $O(\Delta^2)$ initial coloring, and an arbitrary orientation of its edges. This coloring is computed in $\log^* n + O(1)$ rounds using Linial's algorithm \cite{Lin92}. We thus obtain the following theorem. (We note that we put $\oDelta = \Delta$ in the claim of Lemma~\ref{lem:techlemma}, whose claim holds if $\Delta$ is at least a sufficiently large constant. The case of $\Delta=O(1)$ can be handled separately, by first obtaining a $O(\Delta^2)$-coloring of the graph using Linial's algorithm in $O(\log^* n)$ rounds, and then solving the conflict coloring problem in a further $O(\Delta^2) = O(1)$ rounds by greedily assigning in each round colors to all vertices of successive independent sets, corresponding to color classes of the given $O(\Delta^2)$-coloring of $G$.)

\begin{theorem}\label{thm:largeconflict}
There is a local distributed algorithm which solves the $(l,d)$-conflict-coloring problem in $O(\log^* d+ \log^*\Delta ) + \log^* n$ rounds when $\frac{l}{d} \ge 10 \Delta^2 \ln \Delta$.\qed
\end{theorem}

For example, for the special case of list coloring, this gives the following corollary.

\begin{corollary}~\label{cor:listdeltasq}
There is a local distributed algorithm which finds a $(10 \Delta^2 \ln \Delta)$-list-coloring in $\log^* n + O(\log^* \Delta)$ rounds.
\end{corollary}

%We remark that the round complexity of the conflict coloring algorithms is independent of $n$ when a better initial vertex coloring of the graph is given. For example, given a $\Delta^2$-coloring of the graph (which can be obtained using Linial's algorithm~\cite{Lin92} in $\log^* n$ rounds), a $(10\Delta^2 \ln \Delta)$-list-coloring can be found in $O(\log^* \Delta)$ rounds.

In the next section, we will use Theorem~\ref{thm:largeconflict} as a building block for solving conflict coloring instances with a smaller value of ratio $l/d$.

\section{Conflict Coloring with a Small Number of Colors}\label{sec:main}

In this section we show how to apply the techniques from Section~\ref{sec:deltasquare} to obtain a distributed solution to $(l,d)$-conflict coloring problems with $l \geq d\cdot \Delta +1$, such as $(\Delta+1)$-list-coloring.

Whereas we choose to speak of conflict colorings throughout the rest of the paper, we will no longer make use of the general structure of conflict colorings in our technical arguments. The reader focusing on results directly relevant to the $(\Delta+1)$-coloring problem may from now on assume that the problem being solved is $(\Delta+1)$-list-coloring (and specifically, that the conflict degree is $d=1$), and in this context, may rely on Corollary~\ref{cor:listdeltasq} instead of Theorem~\ref{thm:largeconflict} as the relevant ingredient used in the subsequent construction.

In the designed algorithm we will also make use of the following recent result on arbdefective coloring, shown by Barenboim~\cite{Bar15}. For $\beta\geq 0$, a (possibly improper) vertex coloring of a graph $G$ is said to be \emph{$\beta$-arbdefective} if there is an orientation of the edges of $G$ such that, for every node $v$, at most $\beta$ out-neighbors of $v$ have the same color as~$v$.

\begin{lemma}[\cite{Bar15}]\label{lem:arbdefective}
There is a distributed algorithm, parameterized by $k\geq 1$, which, given any graph $G$ with a $\Delta^2$-coloring of its vertex set, produces
for $\beta = O(\frac \Delta k \log \Delta)$ a $\beta$-arbdefective $k$-coloring $V = V_1 \cup \ldots \cup V_k$ of $G$, together with a corresponding orientation of each $G[V_i]$ having outdegree at most $\beta$. The running time of the algorithm is $O(k \log \Delta)$ rounds.
\end{lemma}

Our conflict coloring procedure will assume our graph $G$ is already equipped with a $\Delta^2$-coloring. This can be initially computed using Linial's algorithm~\cite{Lin92}, in $\log^* n$ rounds.

\begin{lemma}
\label{lem:recursive}
Given a $\Delta^2$-vertex coloring of graph $G$ of maximum degree at most $\Delta$, there is an algorithm which solves any conflict-coloring instance on $G$ having conflict degree at most $d$ and color lists $L$ such that $|L(v)|\geq d\cdot \deg_G(v) +1$ for all $v\in V$, in at most $O(\sqrt \Delta \log^{1.5} \Delta (\log \Delta + \log^* d))$ rounds.
\end{lemma}
\begin{proof}
We restrict considerations to the case where $\Delta$ is larger than some fixed constant $\Delta'>0$; otherwise, an appropriate conflict coloring can be obtained in $O(\Delta'^2) = O(1)$ rounds by greedily assigning in each round colors to all vertices of successive independent sets, corresponding to color classes of the given $\Delta^2$-coloring of $G$.

We will design a conflict-coloring procedure $A$, which satisfies the assumptions of the lemma. For a graph $G$, the procedure starts by constructing the $\beta$-arbdefective $k$-coloring $V = V_1 \cup \ldots \cup V_k$ from Lemma~\ref{lem:arbdefective}, for a certain parameter $k$ that will be explicitly stated later. Each of the subgraphs $G[V_i]$ now has an edge orientation with maximum outdegree at most $\beta = \frac \Delta k \log \Delta$, and its vertices are also equipped with locally unique identifiers in the range $\{1,\ldots,\Delta^2\}$ by virtue of the given $\Delta^2$-vertex coloring.

Now, we are ready to solve the conflict-coloring problem on $G$ for a given assignment of lists $L$ such that $|L(v)| \geq d \cdot\deg_G(v)+1$ for all $v\in V$. Our algorithm proceeds in $k$ stages, obtaining in the $i$-th stage a valid (final) conflict-coloring of $G[U_i]$ for a specifically defined subset $U_i \subseteq V_1 \ldots \cup V_i$ (we let $U_0 = \emptyset$), i.e., $\out(v) \in L(v)$ and the color pair $(\out(v),\out(u))$ is not forbidden for the edge $(v,u)$, for all $v \in U_i$, $u \in N_{G[U_i]}(v)$. Let $S^v(u,c_u) \subseteq L(v)$ denote the set of colors available to a node $v$ which are in conflict with a color $c_u$ at neighboring node $u$; we recall that  $|S^v(u,c_u)|\leq d$.  For $i\geq 1$, given a valid conflict-coloring of $G[U_{i-1}]$ at the beginning of the stage, we create for each $v \in V_i$ a list of colors $L'(v) = L(v) \setminus \bigcup_{u \in U_{i-1} \cap N_G(v)} S^v(u,\out(u))$, which may be used at $v$ to extend the conflict coloring of $U_{i-1}$. Now, we use Lemma~\ref{lem:techlemma} to perform a conflict coloring, restricted to color lists $L'$, for the oriented subgraph of $G[V_i]$ induced by those vertices $v\in V_i$, for which the assumptions of the Lemma are satisfied (i.e., $|L'(v)| \ge 10 d \beta^2 \ln \Delta$). This coloring routine takes $O(\log^*\Delta + \log^* d)$ rounds.

We observe that if a vertex $v \in V_i$ is colored during the phase, then it receives a color $\out(v) \in L'(v) \subseteq L(v)$, which does not conflict with the colors of any of its neighbors in $U_{i-1}$ or simultaneously colored vertices from $V_i$; we thus construct $U_i$ by adding to $U_{i-1}$ all vertices colored in the current phase.

If, on the other hand, if vertex $v\in V_i$ does not receive a color, then we must have $|L'(v)| < 10 d \beta^2 \ln \Delta$. By definition, $L'(v)$ consists of the colors in $L(v)$ which are not in conflict with colors chosen in a previous step. For a previously colored neighbor $u \in U_i$, the color $\out(u)$ is in conflict with at most $d$ colors in $L(v)$. Hence, the number of already colored neighbors is $|N_{G[U_{i-1}]}(v)| \geq (|L(v)| - |L'(v)|)/d > \deg_G(v) - 10 \beta^2 \ln \Delta$. In other words, there are at most $10 \beta^2 \ln \Delta$ neighbors of $v$ who did not receive a color yet.

Finally, at the end of the $k$-th stage of the coloring process, we are left with a set $V^* = V \setminus U_k$ of uncolored vertices.

We observe that our conflict-coloring of $G$ can now be completed correctly by conflict-coloring the graph $G^* = G[V^*]$. We define $\Delta^* = 10 \beta^2 \ln \Delta$, having $\Delta^* \geq \Delta_{G^*}$. Moreover, we can complete the conflict-coloring of $G$ by merging the so-far obtained coloring $\out$ on $U_k$ with the conflict-coloring of $G^*$, with inherited conflicting color pairs and color lists $L^*$ defined for $v\in V^*$ as:
$$
L^*(v) = L(v) \setminus \bigcup_{u \in U_{k} \cap N_G(v)} S^v(u,\out(u)).
$$
Since $|L(v)| \geq  d \deg_{G}(v) + 1$ and $U_{k} \cap N_G(v) = \deg_{G}(v) - \deg_{G^*}(v)$, it follows that $|L^*(v)| \geq d \deg_{G^*}(v) + 1$, for all $v\in V^*$. Thus, we may now complete procedure $A$ by recursively applying $A$ to find a list-coloring on $G^*$ with lists $L^*$, and merge the obtained colorings for $U_k$ and $V^*$. By assumption, procedure $A$ on $G^*$ must be given a $(\Delta^*)^2$-vertex coloring of $G^*$, which we can compute using Linial's color reduction mechanism, based on the given $\Delta^2$-coloring of $G$, in $\log^* \Delta$ rounds. Overall, denoting by $T_A(\Delta)$ an upper bound on the running time of algorithm $A$ on a graph of maximum degree at most $\Delta$, we obtain the following bound:
$$T_A(\Delta) \leq O(k \log \Delta) + O(k (\log^* \Delta + \log^* d)) + O(\log^* \Delta) + T_A( O(\beta^2\log \Delta) ),$$
where the first component of the sum comes from the routine of Lemma~\ref{lem:arbdefective}, the second one is the time of the $k$ stages of coloring graphs $G[V_i]$, the third stage is the time of $(\Delta^*)^2$-coloring graph $G^*$, and the final stage comes from the recursive application of procedure $A$. Taking into account that $\beta = O(\frac \Delta k \log \Delta)$, we obtain:
$$T_A(\Delta) \leq O(k (\log \Delta + \log^* d)) + T_A( O(\tfrac {\Delta^2}{k^2} \log^3 \Delta) ).$$
The above expression is minimized for an appropriately chosen (sufficiently large) value $k = O(\sqrt{\Delta \log^3 \Delta})$, for which we eventually obtain $T_A(\Delta) = O(\sqrt \Delta \log^{1.5} \Delta (\log \Delta + \log^* d))$.
\end{proof}

We thus obtain the main result of our paper.

\begin{theorem}\label{theo:main}
There is a distributed algorithm which solves any conflict-coloring instance on $G$ with conflict degree at most $d$ and color lists $L$ such that $|L(v)|\geq d \deg_G(v) +1$ for all $v\in V$, in at most $O(\sqrt \Delta \log^{1.5} \Delta (\log \Delta + \log^* d)) + \log^* n$ rounds.\qed
\end{theorem}

We remark that, for any conflict coloring problem in which the conflict degree $d$ is constant or bounded by any reasonable function of $\Delta$ (i.e., $\log^* d = O(\log \Delta)$), the obtained round complexity simplifies to $O(\sqrt{\Delta}\log^{2.5}\Delta) + \log^* n$. In particular, for the case of $(\Delta+1)$-list-coloring, we have $d =1$, giving the following corollary.

\begin{corollary}
There is a distributed algorithm for the distributed $(\Delta+1)$-list-coloring problem, performing in $O(\sqrt{\Delta}\log^{2.5}\Delta) + \log^* n$ rounds.
\end{corollary}

%%%%%%%%%%%%%%%%%%%%%%%%%%%%%%%%%%%%%%%%%%%%%%%%%%%%%
\section{A Centralized Local Algorithm for Conflict-Coloring}\label{sec:centralized}
%%%%%%%%%%%%%%%%%%%%%%%%%%%%%%%%%%%%%%%%%%%%%%%%%%%%%

In this section, we provide algorithms for solving the conflict coloring problem in the model of centralized local computation. These LCAs are obtained by adapting our distributed algorithms for the \LOCAL model to the centralized local model, using the guidelines in~\cite{ColoringCentralized}. As a special case, we obtain an LCA for $(\Delta+1)$-coloring algorithm with a smaller probe complexity (in terms of $n$ and $\Delta$) than the best previously known approach. Throughout this section we assume a reasonably small conflict degree for the problem (i.e., $\log^* d = O(\log \Delta)$).

\begin{theorem}
There is a deterministic oblivious LCA for solving an instance of $(l,d)$-conflict coloring, satisfying the following:
\begin{itemize}
\item if $l/d\ge10 \Delta^2 \ln \Delta$, then the algorithm performs $\Delta^ {O(\log^ * \Delta)}\log^ * n$ probes per query.
\item if $l/d > \Delta$, then the algorithm performs $\Delta^{O(\sqrt \Delta \log^{2.5} \Delta)} \log^* n$ probes per query.
\end{itemize}
\end{theorem}
	
\begin{proof}
The proof relies on the method from \cite{PR07} for simulating distributed algorithms for the \LOCAL model in the centralized local model (cf.~also~\cite{MISCentralized,ColoringCentralized}). Suppose that we have a distributed local algorithm running in $r$ rounds. We can simulate its execution in the centralized local model with $\Delta^r$ probes, as follows: to answer a query for a node $v$, we probe the whole $r$-neighborhood of $v$, and then run the local algorithm on this neighborhood. Applying this technique directly to a distributed conflict coloring algorithm whose runtime is of the form $O(f(\Delta))+\log^* n$, where $f$ represents some non-decreasing function, we would get an LCA with probe complexity $\Delta^{O(f(\Delta)) + \log^* n}$. To get the $\log^* n$ term out of the exponent, we modify the method in~\cite{ColoringCentralized} a bit. For this purpose, notice that if we assume that we already know a $\Delta^2$-coloring of $G$, then our conflict-coloring algorithms perform in a distributed manner in a number of rounds dependent on $\Delta$, only (cf.~Lemma~\ref{lem:techlemma} and Lemma~\ref{lem:recursive}, respectively, for the two considered cases of the problem). Moreover, there exists an LCA for $\Delta^2$-coloring which performs in $O(\poly(\Delta)) \log^* n$ probes per query, due to~\cite{ColoringCentralized}. We thus propose an LCA for $(l,d)$-conflict-coloring, which, in order to solve a query for a vertex $v$, performs in two phases:
\begin{enumerate}
\item Perform multiple runs of the $\Delta^2$-coloring LCA from~\cite{ColoringCentralized} for queries corresponding to all nodes in the $r$-neighborhood of $v$;
\item Simulate $r$ rounds of a distributed algorithm for $(l,d)$-conflict-coloring for node $v$ using the given $\Delta^2$-coloring of the $r$-neighborhood of $v$.
\end{enumerate}
The first phase requires $\Delta^r \poly(\Delta) \log^* n$ probes of the input graph (i.e., $\poly(\Delta) \log^* n$ probes for each of the $\Delta^r$ queries pased to the $\Delta^2$-coloring LCA), while the second phase does not require any additional probes. To be able to run the $(l,d)$-conflict-coloring algorithm on the $r$-neighborhood, for the case $l/d \ge 10 \Delta^ 2 \ln \Delta$, we set $r = c \log^* \Delta$, for some sufficiently large positive constant $c$ (cf.~Lemma~\ref{lem:techlemma}). This yields an LCA performing $\Delta^ {O(\log^ * \Delta)}\log^ * n$ probes per query.

We apply essentially the same method for the case $l/d > \Delta$, putting $r = c\sqrt \Delta \log^{2.5} \Delta$ for some sufficiently large positive constant $c$ (cf.~Lemma~\ref{lem:recursive}). We obtain an LCA performing $\Delta^{O(\sqrt \Delta \log^{2.5} \Delta)} \log^* n$ probes per query.
\end{proof}

Considering list-coloring as a special case of conflict-coloring, we get the following corollary.

\begin{corollary}
There is a deterministic oblivious LCA for list-coloring, which runs in $\Delta^ {O(\log^ * \Delta)}\log^ * n$ probes per query when all color lists are of length at least $10 \Delta^2 \ln \Delta$, and in $\Delta^{O(\sqrt \Delta \log^{2.5} \Delta)} \log^* n$ probes per query when all color lists are of length at least $\Delta+1$.
\end{corollary}

\section{Conclusion}

This paper presents the problem of $(l,d)$-conflict-coloring in a twofold light. First of all, we show that it is a generalization of numerous symmetry-breaking tasks, which can be solved efficiently in a distributed setting. Secondly, we rely on conflict coloring as a tool to describe intermediate instances of tasks when applying the simplification technique used in our algorithms (cf.~Lemma~\ref{lem:simple}).  In view of our results, the deterministic round complexities of $(\Delta+1)$-coloring, $(\Delta+1)$-list-coloring, and $(l,d)$-conflict-coloring with $l/d>\Delta$, all collapse to $\Otilde (\sqrt \Delta) + \log^* n$ rounds. The sufficiently large value of the ratio $l/d$ in the conflict coloring formulation appears to be what sets these problems apart from not easier tasks, such as MIS, for which no approaches with deterministic $o(\Delta) + \log^* n$ runtime are currently known.

We close the paper by remarking briefly on practical aspects, related to the amount of local computations which individual nodes need to perform to run the proposed algorithms. The most computationally-intensive steps are related to Lemma~\ref{lem:greedy}, which relies on an enumeration of a potentially large set of inputs $I$ to perform a color assignment to each element of the set. The size of this set $I$, and consequently the complexity of local computations of our algorithms, can be bounded as $2^{\Delta^{O(\log^*\Delta)}}$. This value is polynomially bounded with respect to $n$ for values of $\Delta = (\log n)^{o(1/\log^* n)}$. Since the enumeration of set $I$ is the only bottleneck in our algorithms, there exist several ways of speeding up local computations. For example, one can introduce into the algorithms an element of non-uniformity with respect to maximum degree $\Delta$, and for a given upper bound on $\Delta$, construct the solution in Lemma~\ref{lem:greedy} through a pre-computed hash function on set $I$, known to the algorithm, rather than a greedy color selection algorithm. This reduces the local computation time of our algorithms to $\Delta^{O(\log^*\Delta)}$, while preserving the same asymptotic bounds on the round complexity. In the context of LCA's discussed in Section~\ref{sec:centralized}, the cost of local computations in the approach is comparable to its probe complexity. In this sense, our algorithms may be considered satisfactory from a practical perspective in almost the entire range of $\Delta$ sub-polynomial in $n$, which is naturally the main area of focus.

%%%%%%%%%%%%%%%%%%%%%%%%%%%%%%%%%%%%%%%%%%%%%%%%%%%%%

%%%%%%%%%%%%%%%%%%%%%%%%%%%%%%%%%%%%%%%%%%%%%%%%%%%%%

\section*{Acknowledgements}

We have revised Lemma~\ref{lem:techlemma} and its proof following comments from Michael Elkin and Mohsen Ghaffari.

\appendix
\newpage
\section{Remark on Conflict Coloring Formulations for MIS}

The Proposition below shows that there does not exist a $(l,d)$-conflict coloring formulation of MIS with a ratio $l/d>2$, which can be decoded by nodes into a valid MIS by a deterministic local algorithm without subsequent communication. The argument is laid out for the trivial case of a star, i.e., for a tree of diameter 2, and is intended mainly to highlight the general point that the constraints of the MIS problem cannot be conveniently expressed through sets of constraints on individual edges of the graph.

\begin{proposition}
Suppose graph $G$ is a star and consider any instance of $(l,d)$-conflict-coloring over color set $\caC$ on $G$. If there exists a function $f : \caC \to \{0,1\}$, such that a solution $c : V \to \caC$ to the considered conflict coloring problem is valid if and only if $\{v \in V: f(c(v))=1\}$ is a MIS on $G$, then $l/d\leq 2$.
\end{proposition}
\begin{proof}
Let $L(v)$ be the list of colors allowed for a vertex $v\in V$ in the considered conflict coloring instance on the star. Let $L_1(v) \subseteq L(v)$ be the set of all colors $a \in L(v)$ such that $f(a)=1$ and color $a$ may be assigned to vertex $v$ in at least one valid solution to the considered conflict coloring instance, and let $L_0(v) = L(v) \setminus L_1(v)$. Fix $r$ to be the central vertex of the star. Since each of the two possible MIS's on the star must correspond to some solution to the considered conflict coloring problem, we have $L_0(r) \neq \emptyset$ and $L_1(r) \neq \emptyset$. A conflict must exist between each color of $L_1(r)$ and each color of $L_1(u)$, for all $u \neq r$, since otherwise one could extend some conflict coloring $c$ of $G\setminus\{u\}$, for which $f(c(r))=1$, in such a way that $f(c(u))=1$, which does not correspond to a valid MIS. It follows that $d\geq \max_{u \in V\setminus\{r\}}|L_1(u)|$. Moreover, for each color $a \in L_0(r)$, there must exist a vertex $w \neq r$ such that for the edge $\{r,w\}$, color $a$ at vertex $r$ is in conflict with all colors $b\in L_0(w)$ at vertex $w$; otherwise, we could construct a valid conflict coloring in which $c(r)=a$ and $c(w)=b$. This would be a contradiction since neither $w$ nor its only neighbor $r$ would not be in the corresponding MIS because $f(c(r)) = f(c(w)) = 0$. It follows that $d\geq L_0(w)$. By combining the last two observations, we have $2d \geq L_0(w) + \max_{u \in V\setminus\{r\}}L_1(u) \geq L_0(w) + L_1(w) \geq l$, which gives the claim.
\end{proof}

\section{Proof of Lemma~\ref{lem:simple}}\label{sec:lemproof}

We construct instance $P_{i+1} = (L_{i+1}, F_{i+1})$ over color set $\caC_{i+1}$ from instance $P_i = (L_i, F_i)$ over color set $\caC_i$ as follows. We define the color set $\caC_{i+1}$ as the collection of all the subsets of size $k_i = \lfloor \frac{l_i}{e^2 d_i \oDelta} \rfloor$ of $\caC_i$. For each node $v$, we will now appropriately define its color list $L_{i+1}(v)\subseteq {L_i(v) \choose k_i}$ by selecting into $L_{i+1}(v)$ a constant proportion of all $k_i$-element-subsets of $L_i(v)$. The adopted value of parameter $k_i$ is the result of a certain tradeoff: increasing $k_i$ further would indeed increase the list length $l_{i+1}$, but would also result in an explosion of the number of conflicts $d_{i+1}$ (the ratio $l_{i+1}/d_{i+1}$ needs to be controlled in view of Clause~3). The details of the construction of lists $L_{i+1}$ are deferred until later in the proof.

Next, let $\tau_i = \lfloor \frac {k_i} \oDelta \rfloor -1$ be a threshold parameter, which we will use to define the edge set of the conflict graph $F_{i+1}$. For a pair of neighboring nodes $\{u,v\} \in E$, we denote by $S_i^u(v, c_v)$ the set of all colors at vertex $u$ in conflict with color $c_v$ at vertex $v$ in problem $P_i$. We now define the following symmetric conflict relation $(\sim)$ on $V \times \caC_{i+1}$ for the problem $P_{i+1}$:
$$
(u, C_u) \sim (v, C_v) \Leftrightarrow \left\{
\begin{array}{r}
\left|C_u \cap \bigcup_{c_v \in C_v}S_i^u (v, c_v)\right| > \tau_i \\
\text{ or \ } \left|C_v \cap \bigcup_{c_u \in C_u}S_i^v (u, c_u)\right| > \tau_i
\end{array} \right.
$$

%\{(u, C_u), (v, C_v)\} \in E(F_{i+1})
When looking a the left-hand-side of the above relation, it is convenient to think of $C_u$ and $C_v$ as candidates for color values, which are being considered for inclusion in the lists $L_{i+1}(u)$ and $L_{i+1}(v)$ of nodes $u$ and $v$, respectively, in problem $P_{i+1}$. When looking at the right-hand side, we treat $C_u$ and $C_v$ as sets of colors with respect to problem $P_i$. Subsequently, when defining the color lists in problem $P_{i+1}$, we will eliminate those configurations of candidates which generate too many conflicts in node neighborhoods in problem $P_i$.

The above relation, when restricted to permissible vertex colors, defines conflict edges for $F_{i+1}$: given colors $C_u \in L_{i+1}(u)$ and $C_v \in L_{i+1}(v)$ (where we recall that $C_u \subseteq L_i(u)$ and $C_v \subseteq L_i(v)$), we put:
\begin{equation}
\label{eq:defef}
\{(u, C_u), (v, C_v)\} \in E(F_{i+1}) \iff (u, C_u) \sim (v, C_v).
\end{equation}

For this definition of the edge set of $F_{i+1}$, we immediately show how to convert any valid solution to $P_{i+1}$ into a solution for $P_i$ in a single communication round. Indeed, observe that if a node $v$ knows its output $\out_{i+1}(v)$ for $P_{i+1}$ and the outputs of all its out-neighbors in the considered orientation, then it can obtain a valid color in $P_i$ by returning an arbitrary element of the set $\out_{i+1}(v)$ which does not conflict with any of the colors belonging to the corresponding sets of its out-neighbors:
\begin{equation}
\label{eq:down}
\out_i (v) \in \out_{i+1} (v) \setminus \bigcup_{u \in \vec N_G(v)}\bigcup_{c_u \in \out_{i+1}(u)} S_i^v (u, c_u).
\end{equation}
Since, by assumption, the considered solution to $P_{i+1}$ was correct, we have $(u, \out_{i+1}(u)) \not\sim (v, \out_{i+1}(v))$. It follows from the definition of relation $(\sim)$ that in the right-hand-side of expression~\eqref{eq:down}, each element of the union over $u\in \vec N_G(v)$ eliminates at most $\tau_i$ elements from the set $\out_{i+1} (v)$. Moreover, since we have $|\out_{i+1} (v)| = k_i \geq \oDelta \tau_i + 1$, the set from which we are choosing $\out_i (v)$ is always non-empty. Finally, the construction of~\eqref{eq:down} is such that color $\out_i (v)$ cannot conflict with any other color assigned to any of its neighbors in the obtained solution to $P_i$. Thus the obtained solution to $P_i$ is conflict-free with respect to every edge of $G$, which completes the proof of Clause~2 of the Lemma.

In the rest of the construction, we focus on a careful construction of color lists $L_{i+1}(v)\subseteq {L_i(v) \choose k_i}$, so as to ensure the local constructibility of the input instance to $P_{i+1}$ in a single round (Clause~1) and a sufficiently large ratio $l_{i+1}/d_{i+1}$ (Clause~3). The value of $d_{i+1}$ will be fixed as:
$$
d_{i+1}:= 8 \Delta {k_i d_i \choose \tau_i} {l_i \choose k_i - \tau_i}.
$$

We will proceed with the construction of lists $L_{i+1}$ by including all $k_i$-element subsets of $L_i(v)$ in $L_{i+1}(v)$, and then we eliminate some colors from $L_{i+1}(v)$ which would generate too many conflicts in $P_{i+1}$ with any of the possible colors for neighbors $u \in N_G(v)$. Formally, for all $v \in V$, we set:
\begin{align}
\label{eq:lione}
D_{i,v}(u) &:=  \left\{C_v : |\{C_u : (u, C_u) \sim (v, C_v)\}| > \frac{d_{i+1}}{2}\right\} \\
L_{i+1}(v) &:= {L_i(v) \choose k_i} \setminus \bigcup_{u \in N_G(v)} D_{i,v}(u)
\end{align}
The above setting guarantees that the conflict degree bound of $d_{i+1}$ is indeed satisfied by problem $P_{i+1}$. We now show that the condition $|L_{i+1}(v)| \geq \frac{1}{2}{l_i \choose k_i}$ is met for all vertices. To lower bound the size of $L_{i+1}(v)$, we will prove that for each neighbor $u$ of a node $v$, at most $\frac 1 {2 \Delta} {l_i \choose k_i}$ subsets are removed from $L_{i+1}(v)$ when considering conflicts between $u$ and $v$.

\begin{claim}
For any $v \in V$ and $u \in N_G(v)$, we have:
%$$
%\left|\left\{C_v \in {L_i(v) \choose k_i} \ :\  |\{C_u \in {\textstyle{ L_i(u) \choose k_i}} : (u, C_u) \sim (v, C_v)\}| > \tfrac{d_{i+1}}{2}\right\}\right| \leq \frac 1 {2 \Delta} {l_i \choose k_i}.
%$$
\[
|D_{i,v}(u)| \leq \frac 1 {2 \Delta} {l_i \choose k_i}
\]
\end{claim}

\begin{proof}
Consider the bipartite graph with vertex partition $A_v \cup A_u$, where $A_v = \{(v, C_v) : C_v \in {L_i(v) \choose k_i}\}$ and $A_u = \{(u, C_u) : C_u \in {L_i(u) \choose k_i}\}$, and a set of edges $E_\sim$ defined by the conflict relation $(u, C_u) \sim (v, C_v)$ on its nodes. Our goal is to bound the number of vertices in partition $A_v$ having degree at least $\tfrac{d_{i+1}}{2}$ with respect to $E_\sim$ . We will first bound the number of edges in $E_\sim$ as follows. For a fixed set $C_u \in {L_i(u) \choose k_i}$, we bound the number $x_1$ of sets $C_v \in {L_i(v) \choose k_i}$ satisfying the first of the conditions which appear in the definition of relation $(\sim)$:
\begin{equation}
\label{eq:sim1}
\left|C_u \cap \bigcup_{c_v \in C_v}S_i^u (v, c_v)\right| > \tau_i.
\end{equation}
Taking into account that $P_i$ is an instance of conflict coloring with conflict degree at most $d_i$, for any color $c_v$ at $v$ we have $S_i^u (v, c_v)$, and so $|\bigcup_{c_v \in C_v}S_i^u (v, c_v)| \leq \sum_{c_v \in C_v} d_i = k_i d_i$. It follows that $x_1$ can be bounded by the following expression:
$$
x_1 \leq {k_i d_i \choose \tau_i} {l_i \choose k_i - \tau_i} = \frac 1 {8\Delta} d_{i+1}.
$$
Thus, overall, the number of edges of $E_\sim$ satisfying Eq.~\eqref{eq:sim1} is at most $x_1 |A_u| \leq \frac 1 {8\Delta} d_{i+1} {l_i \choose k_i}$. By a symmetric argument, the number of edges contributed by the other condition in the definition of relation $(\sim)$ (i.e., $\left|C_v \cap \bigcup_{c_u \in C_u}S_i^v (u, c_u)\right| > \tau_i$), is also $\frac 1 {8\Delta} d_{i+1} {l_i \choose k_i}$. Overall, we have:
$$
|E_\sim| \leq  \frac 1 {4\Delta} d_{i+1} {l_i \choose k_i}.
$$
The average degree $\delta_\sim$ of a node in $A_v$ with respect to $E_\sim$ is thus bounded by $\delta_\sim \leq \frac 1 {4\Delta} d_{i+1}$. Since only at most $\frac{|A_v|}{2\Delta} = \frac 1 {2 \Delta} {l_i \choose k_i}$ nodes in $A_v$ can have a degree higher than $2\Delta \delta_\sim \leq \frac{d_{i+1}}{2}$, the claim follows.
\end{proof}
As a direct corollary of the above claim and of the definition of $L_{i+1}(v)$ in~\eqref{eq:lione}, we have obtained the sought bound $|L_{i+1}(v)| \geq \frac 1 2 {l_i \choose k_i}$. Formally, to guarantee that $P_{i+1}$ is an instance of a $(l_{i+1}, d_{i+1})$-conflict-coloring problem with lists of size precisely equal to:
$$l_{i+1}:= \frac 1 2 {l_i \choose k_i},
$$
in the case when the size of some $L_{i+1}(v)$ still exceeds  $l_{i+1}$, node $v$ removes arbitrarily some elements of $L_{i+1}(v)$ so that its size becomes exactly $l_{i+1}$. Bearing in mind the description of color lists $L_{i+1}$ according to Eq.~\eqref{eq:lione} and the edges of the conflict graph $F_{i+1}$ according to Eq.~\eqref{eq:defef}, a single-round distributed algorithm for computing an instance of $P_{i+1}$ based on an instance of $P_i$ follows directly from the construction. This completes the proof of Clause~1 of the Lemma.

Finally, we complete the proof of the lemma with the following claim, which shows that Clause~$3$ is also satisfied.

\begin{claim}
\label{claim:ratioineq}
For any $\eps>0$, the following inequality holds when $\oDelta$ is at least a sufficiently large constant: $$\frac {l_{i+1}}{d_{i+1}} > % \frac{1} {16\Delta} e^{\tau_i} =
\frac 1 {\Delta} \exp\left(\frac{1}{(e^2 + \eps)\oDelta^2}\frac{l_i}{d_i}\right).$$
\end{claim}

\begin{proof}
Using the inequality $n! \geq \left(\frac{n}{e}\right)^n$, and the definitions of $l_{i+1}$, $d_{i+1}$, $k_i$, and $\tau_i$, we get:
\begin{eqnarray}
\frac {l_{i+1}}{d_{i+1}} &= &\frac {{l_i \choose k_i} / 2} {8 \Delta {k_i d_i \choose \tau_i} {l_i \choose {k_i - \tau_i}}} \nonumber\\
&= & \frac 1 {16 \Delta} \frac{ \tau_i! (k_i d_i- \tau_i)! (k_i- \tau_i)! (l_i-k_i+ \tau_i)!}{k_i! (l_i-k_i)! (k_i d_i)!} \nonumber\\
&\geq& \frac 1 {16 \Delta} \frac{(l_i-k_i)^ {\tau_i}  \tau_i! }{k_i^ {\tau_i} (k_i d_i)^ {\tau_i}} %\\& \geq &
\geq \frac 1 {16 \Delta}\left(\frac{(l_i-k_i)  \tau_i}{ek_i^2d_i} \right)^ {\tau_i} \label{eq:newratio}%\\& \geq &
\end{eqnarray}
Taking into account that $\tau_i = \lfloor\frac {k_i} \oDelta \rfloor -1 \geq \frac {k_i} \oDelta - 2$, $k_i = \lfloor \frac{l_i}{e^2 d_i \oDelta}\rfloor \geq \frac{l_i}{e^2 d_i \oDelta} - 1$, and so $l_i \geq e^2 d_i \oDelta k_i$, we can lower-bound the base of the last expression in \eqref{eq:newratio} as:
\begin{align*}
\frac{(l_i-k_i)  \tau_i}{ek_i^2d_i} &\geq \frac{(e^2 d_i \oDelta k_i -k_i) (\frac {k_i} \oDelta - 2)}{ek_i^2d_i} \\
&= e\left(1 - \frac{1}{e^2 d_i\oDelta}\right) \left(1 - \frac{2\oDelta}{k_i}\right) \\ &> e\left( 1- \frac{1}{\oDelta} -  \frac{2\oDelta}{k_i}\right).
\end{align*}
In what follows, we assume that $\frac {l_i}{d_i} > e^2 \oDelta^2 \ln \Delta \geq e^2 \oDelta^2 \ln \oDelta$; otherwise, the claim of the lemma is trivially true (since $\frac{l_{i+1}}{d_{i+1}} \geq 1$ always holds). We obtain that for sufficiently large $\oDelta$, $k_i \geq \oDelta \ln \oDelta - 1 > \frac 1 2 \oDelta \ln \oDelta$, and so:
\begin{align*}
\frac{(l_i-k_i)  \tau_i}{ek_i^2d_i} &> e\left( 1- \frac{1}{\oDelta} -  \frac{2\oDelta}{k_i}\right) > e\left( 1- \frac{1}{\oDelta} -  \frac{4}{\ln \oDelta}\right) \\ &> e\left( 1- \frac{5}{\ln \oDelta}\right)  > \exp(1 - \eps/10),
\end{align*}
where the last inequality holds for $\oDelta$ sufficiently large with respect to $\eps$. Now, taking into account that $\tau_i \geq \frac {k_i} \oDelta - 2 \geq \frac{l_i}{e^2 d_i \oDelta^2} - 3$, we obtain from Eq.~\eqref{eq:newratio}:
\begin{align*}
\frac {l_{i+1}}{d_{i+1}} &> \frac 1 {16 \Delta} \exp\left[(1 - \eps/10)\left(\frac{l_i}{e^2 d_i \oDelta^2} - 3\right)\right] \\ &> \frac 1 {\Delta} \exp\left(\frac{1}{(e^2 + \eps)\oDelta^2}\frac{l_i}{d_i}\right),
\end{align*}
where again the last bound holds for $\oDelta$ sufficiently large with respect to $\eps$. This completes the proof of the claim.
\end{proof}

%%%%%%%%%%%%%%%%%%%%%%%%%%%%%%%%%%%%%%%%%%%%%%%%%%%%%
\end{document}